\documentclass[submission,copyright,creativecommons]{eptcs}

\usepackage{tikz}
\usepackage{caption}
\usepackage{subfigure}
\usepackage{ wasysym }
\usetikzlibrary{fit}
\usepackage{float}
\usepackage{verbatim}

\usepackage{amsthm}
\newtheorem{definition}{Definition}
\newtheorem{corollary}{Corollary}
\newtheorem{theorem}{Theorem}

\newtheorem{lemma}{Lemma}
\newtheorem{proposition}{Proposition}
\newtheorem{observation}{Observation}
\newtheorem{notation}{Notation}
\newtheorem{example}{Example}[section]
\providecommand{\event}{TARK 2019} 
\usepackage{underscore}           

\def\titlerunning{Aggregation in Value-Based Argumentation Frameworks}
\def\authorrunning{G. Lisowski, S. Doutre \& U. Grandi}
\author{Grzegorz Lisowski
\institute{University of Warwick\\ United Kingdom}
\and
Sylvie Doutre
\institute{University of Toulouse\\ France}
\and
Umberto Grandi
\institute{University of Toulouse \\ France}
}

\title{Aggregation in Value-Based Argumentation Frameworks}

\begin{document}
\maketitle
\begin{abstract}
  Value-based argumentation enhances a classical abstract argumentation graph - in which arguments are modelled as nodes connected by directed arrows called attacks - with labels on arguments, called \emph{values}, and an ordering on values, called \emph{audience}, to provide a more fine-grained justification of the attack relation. With more than one agent facing such an argumentation problem, agents may differ in their 
  ranking of values. When needing to reach a collective view, such agents face a dilemma between two equally justifiable approaches: aggregating their views at the level of values, or aggregating their attack relations, remaining therefore at the level of the graphs. We explore the strenghts and limitations of both approaches, employing techniques from preference aggregation and graph aggregation, and propose a third possibility aggregating rankings extracted from given attack relations.
\end{abstract}

\section{Introduction}

The strength of arguments plays an important role in establishing the outcome of a discussion. Strong arguments have a stronger impact on the interlocutors than arguments perceived as weak. In particular, a strong argument can be accepted
even if it is undermined by insignificant issues. The methods of formal modelling of the impact of the arguments' strength have been intensively studied in the abstract argumentation literature (see e.g. \cite{StrangeNumber,dunne2011weighted}). In particular, the problem of the perceived strength of arguments is crucial in a multi-agent environment. 

It is far from being clear how a consensus can be reached among multiple agents 
disagreeing about which arguments are stronger than others. This observation connects to an extensive discussion on methods of finding a consensus among agents who disagree about the aspects of argumentation they participate in (e.g. \cite{dunne2012argument,awad2015judgement,delobelle2016merging}).

One of the approaches to this problem emphasizes the role of \emph{values} to which arguments appeal in determining their strength. In value-based argumentation, pioneered by Bench-Capon \cite{bench2003persuasion} and whose motivation traces back to Perelman \cite{perelman1971new}, the concept of abstract argumentation due to Dung \cite{dung1995acceptability} is extended to incorporate the information on the values associated with the particular arguments. More precisely, following Dung's model, arguments are conceived as points in the graph, linked by an \emph{attack} relation, which allows for designing rationality constraints (semantics) on acceptable sets of arguments. Additionally, each argument is assigned a unique value\footnote{The generalisation of the model by allowing arguments to appeal to multiple values has been studied by, e.g., Kaci and van der Torre \cite{kaci2008preference}. For the sake of simplicity, in the current paper we will work with the simple scenario in which an argument is assigned a unique value.}. Further, in this approach, the strength of arguments is induced by an agent's preference over values that arguments appeal to: an attack from an argument appealing to a weak value on an highly valued argument is blocked. In this way, a \emph{defeat graph} based on an individual view on the hierarchy of values is induced. The value-based argumentation approach has been further studied and extended in various ways (e.g. \cite{bench2007audiences}). 

In this paper we consider a value-based argumentation framework, and several agents, each of them with its own preference over values. In this multi-agent context, a consensus is sought. Observe that the considered consensus is not directly about the choice of acceptable arguments, what is known in the literature as the semantics \cite{dung1995acceptability}, but rather about the attack graph itself and its justification in terms of preference over values.

Note that there are two intuitive manners of aggregating the possibly conflicting views of the agents.
The first method focuses on the individual ordering over values, and aggregates them into a collective ordering over values that in turns induces an attack relation over labelled arguments. Off-the-shelf techniques from \emph{preference aggregation} and social choice theory in general can be used to find suitable aggregators 
(see, e.g., \cite{arrow1952social,satterthwaite1975strategy,gibbard1973manipulation,brandt2016handbook}).
The second approach to this problem is to aggregate directly the agents' attack graphs, searching for a collective justification in the form of an ordering over the values a posteriori. Techniques in
\emph{graph aggregation} 
\cite{endriss2017graph} and \emph{belief merging} \cite{EveraereEtAlTRENDS2017} have already proven useful in various settings related to argumentation 
\cite{delobelle2016merging, delobelle2018, chen2019preservation}.

The aim of the current paper is to investigate the properties of the two approaches 
%
in the context of value-based argumentation, their benefits and limitations. 
The two possibly conflicting approaches are reminescent of the discursive dilemma in the literature on judgment aggregation \cite{Endriss16,GrossiPigozzi14}, where aggregating the views of a group of agents on the premises supporting a given statement might conflict with the aggregated view on the statement itself, posing a serious challenge to the search for justifications in collective reasoning.
This point relates also to the recent line of research on explainability of automated decision-making, which deals with designing methods to explain to a human user automated decisions 
(see, e.g. \cite{miller2019}). Argumentation theory is already 
being employed in this context (e.g \cite{Fan2015OnCE}).

\paragraph{Our contribution} After setting the stage for aggregating the views of agents on values and attacks over arguments, we 
analyze the arguably simpler method of aggregating directly the individual attack relations. As a first impossibility result we show that, if a number of natural properties are expected for the graph aggregation rule, then it is not always possible to extract a collective justification for the aggregated attack relation. 
We then move to the study of the aggregation of orderings over values. As there are several orderings over values that can justify a given attack relation,  we investigate whether the result of the aggregation is stable with respect to this choice. We show a second impossibility result, again under standard axiomatic requirements on the aggregation procedure. As a side result, we characterize frameworks in which the ordering over values justifying a given attack relation is unique.
We conclude then by providing an alternative aggregation procedure lying between the two studied approaches. To do so, we make use of the value-based framework, the context of our problem, to define an aggregation procedure over preferences on values that are suitably extracted from the individual attack relations, showing that it can ensure the same axiomatic properties of the preference aggregation rule chosen.

\paragraph{Related Work}
%
It is worth noting that the methods of obtaining a consensus structure of argumentation based on values has been studied before. For instance, Modgil \cite{modgil2009reasoning} explored a modes for meta-argumentation about the conflicts between arguments which captures also argumentation about preferences over values.
Furthermore, Airiau~\emph{et al.}\cite{airiau2016rationalisation} studied methods of checking if a set of argumentation frameworks can be seen as a set of defeat graphs reflecting disagreement on the importance of values. This is the closest related work to this paper, and we will refer to it further in the following sections. 
Finally, preference aggregation in the context of value-based argumentation has been previously proposed in the context of value-based argumentation by Pu~\emph{et al.}\cite{pu2013social}.

Outside the realm of value-based argumentation, a number of approaches have recently considered the problem of aggregation in multiagent argumentation. Coste-Marquis~\emph{et al.} \cite{Coste-MarquisAIJ07} were the first to tackle this problem, proposing distance-based methods for the aggregation of argumentation structures. Dunne \emph{et al.} \cite{DunneMW12} later proposed an axiomatic study of the aggregation of argumentation frameworks, later expanded by the work of Delobelle \emph{et al.} \cite{delobelle2018}. Closest to our analysis are the works of Awad \emph{et al.} \cite{awad2015judgement} and Chen and Endriss \cite{chen2019preservation}, which focuses on several problems related to the collective rationality of the aggregation process, i.e., whether properties satisfied by the input attack relations are preserved in the collective outcome.

This paper expands previous work by the authors \cite{SAFA18}, building on the Master thesis of Lisowski~\cite{GL18}. 

\paragraph{Paper structure} The paper is organized as follows. In Section~2 we give the basic definitions of value-based argumentation, preference and graph aggregation, and list a number of desirable properties for the aggregation process.
Further, Section~3 describes the properties of the approaches based on aggregating submitted graphs and on preference orderings. In Section~4, we analyze the combined approach. Finally, in Section~5 we provide conclusions and suggestions for further development.

\section{The setting}

The model we study in this paper is the \emph{value-based argumentation setting}, due to Bench-Capon \cite{bench2003persuasion}. This formal framework extends Dung's abstract argumentation model \cite{dung1995acceptability}. Hence, the basic notion employed in the studied setting is of an \emph{argumentation framework}. It is a directed graph, in which vertices correspond to arguments, and edges to the attack relation, capturing conflicts between arguments.

\begin{definition}
		An \emph{argumentation framework} ($\mathit{AF}$) is a pair $\mathit{AF}= \langle A, \rightarrow \rangle$, where $A$ is a set of arguments and $\rightarrow \subseteq A^2$ is the attack relation. We denote the fact that $\langle a,b \rangle \in \rightarrow$ as $a \rightarrow b$.  $\langle a,b \rangle \not \in \rightarrow$ is denoted $a \not\rightarrow b$.
\end{definition}
	
In order to capture how arguments appeal to values, argumentation frameworks are extended to \emph{value-based argumentation frameworks}. These are labeled directed graphs, in which labels correspond to values that arguments relate to, an argument relating to one value only. For the sake of simplicity
we assume that every argument is associated with a single value. 
	
\begin{definition}
		A \emph{value-based argumentation framework} (\textit{VAF}) is a tuple $\textit{VAF} = \langle A, \rightarrow, V, val\rangle$,
		where $A$ is a set of arguments, $\rightarrow \subseteq A^2$ is an attack relation, $V$ is a set of values and $val: A \rightarrow V$ is a function assigning values to arguments.
\end{definition}
	
	Then, we can define how an agent's preferences over values determine the \emph{strength of values} from the perspective of a particular agent. To achieve that, agents are allowed to express preference orderings over values. Following the literature on value-based argumentation, such a preference ordering is called an \emph{audience}.
	
	\begin{definition}
		Let $\textit{VAF} = \langle A, \rightarrow, V, val\rangle $. An \emph{audience} $P$ is a linear ordering\footnote{A linear ordering is an irreflexive, transitive and complete binary relation.} over $V$. We denote that a value $v_1\in V$ is more preferable than a value $v_2\in V$ for $P$ as $ v_1 \succ_P v_2$. 
	\end{definition}
	
	Then, the way in which agents perceive the strength of arguments influences the way in which they perceive the structure of argumentation. Intuitively, an agent might disregard the fact that a weak argument undermines a strong argument, as it is not important enough to be considered a legitimate reason to reject a vital point. Formally, we say that an argument $a$ \emph{defeats} an argument $b$, if $a$ attacks $b$ and either $b$'s value is weaker than $a$'s, or $a$ and $b$ appeal to the same value.
    It is worth noting that as we require audiences to be linear, this means that $a$ must be stronger than $b$.
	
		\begin{definition}
		Let $\textit{VAF} = \langle A, \rightarrow, V, val\rangle $ be a \textit{VAF} and $P$ be an audience. Then, we say that an argument $a$ \emph{defeats} an argument $b$ \emph{for $P$} (we denote it as $a \rightarrow_P b$) iff $a \rightarrow b$ and it is not the case that $val(b) \succ_P val(a)$.
	\end{definition}
	
Further, given an initial \textit{VAF} and an audience, we can consider an argumentation framework based on the set of arguments present in \textit{VAF}, and the defeat relation. Such an argumentation framework is called a \emph{defeat graph}.
	
	\begin{definition}
		Let $\textit{VAF} = \langle A, \rightarrow, V, val\rangle $  and $P$ be an audience. The \emph{defeat graph} of $\textit{VAF}$ \emph{induced by $P$} is an argumentation framework $\textit{AF} = \langle A, \rightarrow_P \rangle$.
	\end{definition}
	
	Let us illustrate the defined notion with an example, adapted from \cite{airiau2016rationalisation}:
	
 \begin{example}\label{ex:running}
		Consider a debate regarding the possible ban of diesel cars, aimed at the reduction of air pollution in big cities. The following arguments are included in the discussion:
		\begin{itemize}
			\item A - Diesel cars should be banned.
			\item B - Artisans, who should be protected, cannot change their cars as it would be too expensive for them.
			\item C - We can subsidize electric cars for artisans.
			\item D - Electric cars, which could be a substitute for diesel, require too many new charging stations.
			\item E - We can build some charging stations.
			\item F - We cannot afford any additional costs.
			\item G - Health is more important than economy, so we should spend whatever is needed for fighting pollution.
		\end{itemize}
		
        \noindent
        Further, it can be noticed that these arguments appeal to certain values. In particular, arguments $A,G$ appeal to environmental responsibility (ER), $B,C$ to social fairness (SF), $F$ to economic viability (EV) and $D,E$ - to infrastructure efficiency (IE).
        
		These arguments are represented on the graph with a mapping of values depicted on Figure~\ref{fig:running}. For each argument, the first element of its description is its name, and the second one is the name of the value it appeals to.

		\begin{figure}[H]
			\centering
		\begin{tikzpicture}
			[->,shorten >=1pt,auto,node distance=1.5cm,
			semithick]
			\node[shape=circle,draw=black] (A) {\textit{A}, ER};
			\node[shape=circle,draw=black] (B) [above of=A, left of= A] {\textit{B}, SF};
			\node[shape=circle,draw=black] (D) [below of = A, left of=A] {\textit{D}, IE};
			\node[shape=circle,draw=black] (C) [ left of=B] {\textit{C}, SF};
			\node[shape=circle,draw=black] (E) [ left of=D] {\textit{E}, IE};
			\node[shape=circle,draw=black] (F) [above of=E, left of= E] {\textit{F}, EV};
			\node[shape=circle,draw=black] (G) [left of= F] {\textit{G}, ER};
			
			\path (B) edge (A);
			\path (D) edge (A);
			\path (C) edge (B);
			\path (E) edge (D);
			\path (F) edge (C);
			\path (F) edge (E);
			\path (F) edge [bend right] (G);
			\path (G) edge [bend right] (F);
            
			\end{tikzpicture}
			\caption{Value-based argumentation framework VAF  
			of Example~\ref{ex:running}; 
			each node is labeled with its name and the name of the value it appeals to; the arrows correspond to the attack relation.}\label{fig:running}

		\end{figure}
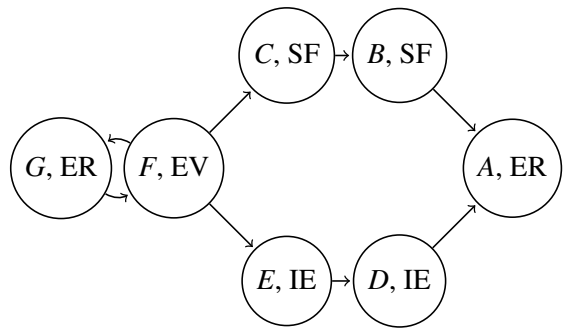
	
	Let us now consider the structure of this discussion from the perspectives of two experts 
	of a decision-making jury, that should decide on whether Diesel cars should be banned or not. 
	
	For Expert~1, economic viability is the most important. She ranks infrastructure efficiency lower, but higher than social fairness.  Environmental responsibility is the least important for her. Then, from her point of view attacks in which the attacker appeals to a less important value than the attacked argument are disregarded. Taking her preferences into account, the  structure presented in Figure~\ref{fig:expert12} (a) is obtained, after the elimination of disregarded attacks.
    Let us now consider another expert
	of the jury, who believes that economic viability is the most important value. Expert~2 ranks environmental responsibility second, and social fairness third. Finally, she considers infrastructure efficiency as the least important. From her perspective, the structure of successful attacks is much different, as indicated in Figure~\ref{fig:expert12} (b). 
		
	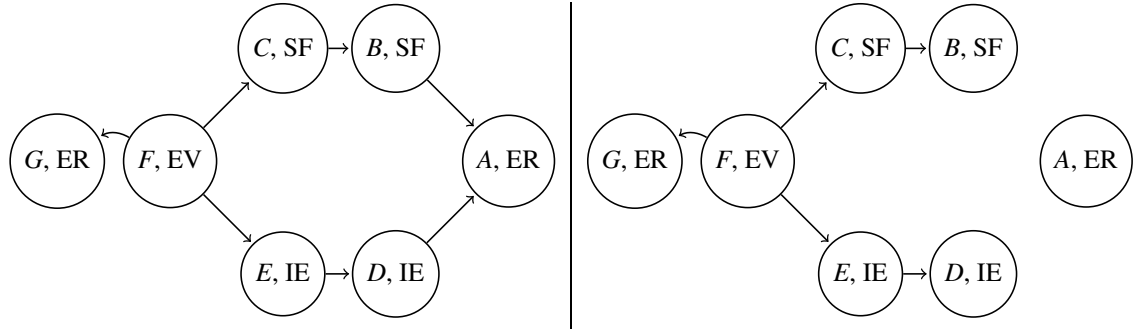
\begin{figure}[H]
	\centering
	\begin{tabular}{c|c}
\centering
			\begin{tikzpicture}
			[->,shorten >=1pt,auto,node distance=1.5cm, semithick]
			\node[shape=circle,draw=black] (A) {\small{\textit{A}, ER}};
			\node[shape=circle,draw=black] (B) [above of=A, left of= A] {\small \textit{B}, SF};
			\node[shape=circle,draw=black] (D) [below of = A, left of=A] {\small \textit{D}, IE};
			\node[shape=circle,draw=black] (C) [ left of=B] {\small \textit{C}, SF};
			\node[shape=circle,draw=black] (E) [ left of=D] {\small \textit{E}, IE};
			\node[shape=circle,draw=black] (F) [above of=E, left of= E] {\small \textit{F}, EV};
			\node[shape=circle,draw=black] (G) [left of= F] {\small \textit{G}, ER};
			
			\path (B) edge (A);
			\path (D) edge (A);
			\path (C) edge (B);
			\path (E) edge (D);
			\path (F) edge (C);
			\path (F) edge (E);
			\path (F) edge [bend right] (G);
			
			\end{tikzpicture}
	     &  
\centering
    
			\begin{tikzpicture}
			[->,shorten >=1pt,auto,node distance=1.5cm,
			semithick]
			\node[shape=circle,draw=black] (A) {\small \textit{A}, ER};
			\node[shape=circle,draw=black] (B) [above of=A, left of= A] {\small \textit{B}, SF};
			\node[shape=circle,draw=black] (D) [below of = A, left of=A] {\small \textit{D}, IE};
			\node[shape=circle,draw=black] (C) [ left of=B] {\small \textit{C}, SF};
			\node[shape=circle,draw=black] (E) [ left of=D] {\small \textit{E}, IE};
			\node[shape=circle,draw=black] (F) [above of=E, left of= E] {\small \textit{F}, EV};
			\node[shape=circle,draw=black] (G) [left of= F] {\small \textit{G}, ER};

			\path (C) edge (B);
			\path (E) edge (D);
			\path (F) edge (C);
			\path (F) edge (E);
			\path (F) edge [bend right] (G);
			
			\end{tikzpicture}
			\end{tabular}

			\caption{Defeat graphs based on (a) Expert~1's  (EV $\succ$ IE $\succ$ SF $\succ$ ER) and (b) Expert~2's (EV $\succ$ ER $\succ$ SF $\succ$ IE) audiences. In the graphs, nodes are labeled with names of arguments and names of values they appeal to. The arrows correspond to the induced defeat relations  based on the experts' audiences.
			}\label{fig:expert12}
		\end{figure}

    \end{example}

It is worth observing that in the current paper we assume for the sake of simplicity that agents have a certain \emph{common ground}. Namely, they agree on the set of arguments and on the values that they appeal to. This assumption is indeed restrictive: in many cases participants of a dispute disagree on these parameters. We aim at addressing these issue in future work.

\subsection{Aggregating preferences}
A natural way  of retrieving a collective view on the structure of argumentation taking into account individual opinions on the importance of values is to employ a \emph{preference aggregation} approach. This method considers a profile of preference orderings, corresponding to individuals' opinions, and provides a single, collective preference ordering. We will denote the set of individuals as $\mathcal{N}= \{1, \dots, n \}$.

	\begin{definition}
	Let $\textit{Pref} = \{\succ_1, \dots, \succ_n \}$ be a set of preference orderings over a set $V$. A profile of \emph{preference orderings} is a tuple $\textbf{P} = \langle \succ_1, \dots, \succ_n \rangle$ consisting of the elements of $\textit{Pref}$.
	\end{definition}
	
\begin{definition}[Preference Aggregation Rule]
Let $V = \{ v_1, \dots , v_n \}$ be a set of options and $\mathcal{P}$ be the set of all  linear orderings over $V$. Then a \emph{preference aggregation rule} is a function $F: \mathcal{P}^m \rightarrow \mathcal{P}$. We denote the set of agents supporting $v_i \succ v_j$  in a profile \textbf{P} as $N_{\textbf{P}}^{v_i \succ v_j}$.
\end{definition}

This method provides a straightforward way of dealing with disagreements with respect to views on the hierarchy of values. 
All agents start first by submitting their preference orderings over values. 
Further, the orderings are aggregated with employment of a chosen preference aggregation function, resulting in a collective preference ordering ($\succ_{coll}$). Finally, by applying this collective preference ordering over values to a considered \textit{VAF}, a collective defeat graph $\textit{AF}_{coll}$ is obtained. Figure~\ref{scheme1} depicts this process. 


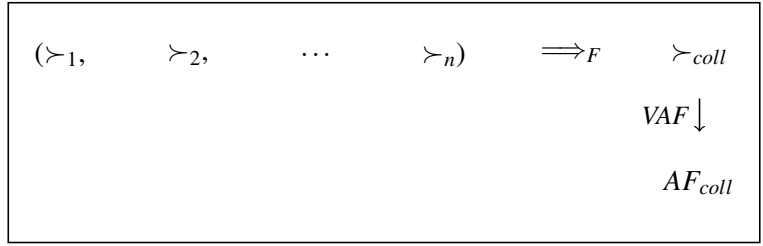
\begin{figure}[H] 
\centering

\begin{tikzpicture}
		[->,auto,node distance=1.7cm,
		semithick]
		\node[shape=circle,draw=white] (A) {($\succ_1$,};
		\node[shape=circle,draw=white] (B) [ right of= A] {$\succ_2$,};
		\node[shape=circle,draw=white] (C) [ right of= B] {$\dots$};
		\node[shape=circle,draw=white] (D) [ right of= C] {$\succ_n$)};
		\node[shape=circle,draw=white] (E) [ right of= D] {$\Longrightarrow_F$};
		\node[shape=circle,draw=white] (F) [ right of= E] {$\succ_{coll}$};
		
		
		\node[shape=circle,draw=white] (S)  [below of= F] {$\textit{AF}_{coll}$};
		
		
		\draw [->] (F) -- (S) node[midway,left] {\small \textit{ VAF}} ;
		
		
		\node[draw=black, fit=(A) (S) ](FIt1) {};
		\end{tikzpicture} 

\caption{Preferences over values are aggregated into a collective ordering on values, which in turn induces an attack graph once a \emph{VAF} is fixed.} \label{scheme1}
\end{figure}

A simple example of a preference aggregation rule, which we will use to illustrate the proposed mechanism, is the \emph{Borda} rule. Let us first define the \emph{rank} of an option from the perspective of the given preference ordering.

\begin{notation}
Let $P$ be a linear order over some set $V$. 
We denote as $rank_{P}(v)$ the position of the option $v$ in the ordering $P$. Formally, $rank_{P}(v) = |\{ v' \in V| v \succ_P v'\}|$. 
\end{notation}

Then, given a profile  \textbf{P} of preference orderings $P_i$ over set $V$, $\textit{Borda}(\textbf{P})=\succ_{coll}$, such that $v_i \succ_{coll} v_j$ iff $\displaystyle \sum_{P_i \in \textbf{P}}  \  rank_{P_i}(v_i) > \sum_{P_i \in \textbf{P}} \ rank_{P_i}(v_j)$ for every $v_i, v_j \in V$, combined with a tie-breaking rule to obtain a strict ranking in case of equality of the Borda score.\footnote{We refer to $\displaystyle \sum_{P_i \in \textbf{P}}  \  rank_{P_i}(v_i)$ as to the \emph{score} of $v_i$.} 
%
Let us illustrate how a collective defeat graph is computed using the Borda rule on our running example.

\begin{example}\label{ex:borda}
(Continuation of Example~\ref{ex:running}) Let us consider an additional expert, Expert~3. Let us present her audience, and let us recall the audiences of the other two experts. These three experts form a  panel \textbf{P}.
\begin{itemize}
    \item Expert 1: EV $\succ$ IE $\succ$ SF $\succ$ ER
    \item Expert 2: EV $\succ$ ER $\succ$ SF $\succ$ IE 
    \item Expert 3: SF $\succ$ ER $\succ$ EV $\succ$ IE
\end{itemize}

\noindent
Let us now calculate the result of the Borda rule for \textbf{P}. The scores are: ER: 4, EV: 7, IE: 2, SF: 5. So, \textit{Borda}(\textbf{P})= EV $\succ$ SF $\succ$ ER $\succ$ IE. The defeat graph for this ordering is presented in Figure~\ref{fig:borda}; this is the collective defeat graph for the panel.

 	\begin{figure}[H]
 
 	    \centering
			\begin{tikzpicture}
			[->,shorten >=1pt,auto,node distance=1.5cm,
			semithick]
			\node[shape=circle,draw=black] (A) {\textit{A}, ER};
			\node[shape=circle,draw=black] (B) [above of=A, left of= A] {\textit{B}, SF};
			\node[shape=circle,draw=black] (D) [below of = A, left of=A] {\textit{D}, IE};
			\node[shape=circle,draw=black] (C) [ left of=B] {\textit{C}, SF};
			\node[shape=circle,draw=black] (E) [ left of=D] {\textit{E}, IE};
			\node[shape=circle,draw=black] (F) [above of=E, left of= E] {\textit{F}, EV};
			\node[shape=circle,draw=black] (G) [left of= F] {\textit{G}, ER};
			
			\path (B) edge (A);
			\path (C) edge (B);
			\path (E) edge (D);
			\path (F) edge (C);
			\path (F) edge (E);
			\path (F) edge [bend right] (G);

			\end{tikzpicture}

			\caption{Collective defeat graph for the panel \textbf{P}, under the Borda rule. }\label{fig:borda}
		\end{figure}
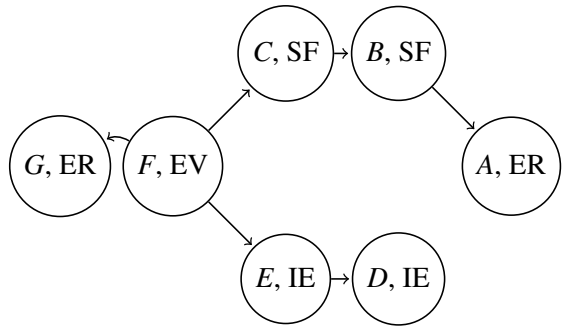
\end{example}

\subsection{Aggregating defeat graphs}

In performing preference aggregation over values, the previous section assumed 
that these individual preferences are known, as an input data of the search for a collective position. However, it may happen that what is known from the agents is the way they see arguments and their relationships, that is, what their defeat graph is, and not what their preferences over values are. For instance, the experts may not want to reveal what their inner preferences are, but just present how they see the resulting situation in terms of arguments and attacks. 
In such contexts, the search for a collective position can be done by aggregating the defeat graphs.

An intuitive requirement for this process is that the resulting collective defeat graph be \emph{justifiable} with respect to a value-based argumentation framework the input defeat graphs are based on: that is, there should exist an ordering of values that, when applied to the initial known or unknown value-based framework, produces the collective defeat graph.

If graph aggregation is fully justified when the preferences over values are not known, one may also think of using this technique to obtain a collective view when they are known. The two techniques, graph aggregation and preference aggregation, will be compared in further sections. For now, let us present graph aggregation. 

A \emph{graph aggregation} rule is a function taking a profile of graphs as an input and providing a single graph. It is worth noting that in the considered setting we only take into account profiles of graphs sharing the set of vertices and require that the collective graph is also based on this set.

\begin{definition}[Graph Aggregation Rule]
Let $A$ be a set of arguments and $Graphs$ be the set of all argumentation frameworks  based on  $A$.  Then a \emph{graph aggregation rule} is a function $F: \textit{Graphs}^m \rightarrow \textit{Graphs}$. For any pair of arguments $a,b \in A$, we denote the set of agents supporting $a \rightarrow b$  in a profile \textbf{AF} as $N_{\textbf{AF}}^{a \rightarrow b}$.
\end{definition}

The aggregation process is depicted on Figure~\ref{scheme2}. 

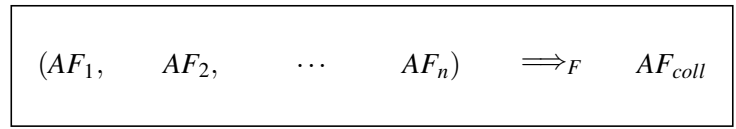
\begin{figure}[H] 
\centering

\begin{tikzpicture}
		[=>,->,shorten >=1pt,auto,node distance=1.6cm,
		semithick]
		\node[shape=circle,draw=white] (A) {$(\textit{AF}_1,$};
		\node[shape=circle,draw=white] (B) [ right of= A] {$\textit{AF}_2,$};
		\node[shape=circle,draw=white] (C) [ right of= B] {$\dots$};
		\node[shape=circle,draw=white] (D) [ right of= C] {$\textit{AF}_n)$};
		\node[shape=circle,draw=white] (E) [ right of= D] {$\Longrightarrow_F$};
		\node[shape=circle,draw=white] (F) [ right of= E] {$\textit{AF}_{coll}$};

		\node[draw=black, fit=(A) (F) ](FIt1) {};
		\end{tikzpicture} 

\caption{Individual graphs are aggregated into a collective graph.}\label{scheme2}
\end{figure}

An example of an intuitive class of graph aggregation rules is the class of \emph{quota rules}. There, an edge is included in the collective graph if a specified fraction of agents includes it in their submitted graphs. 

\begin{definition}[Quota rule]
	Let $\textbf{AF}$ be a profile of argumentation frameworks. Then, $F$ is a \emph{quota rule} if there is $q \in [0,1]$ such that for any attack $a \rightarrow b$, $a \rightarrow b \in F(\textbf{AF})$ iff $N_{\textbf{AF}}^{a \rightarrow b} \geq \lfloor q* n \rfloor$, where $n$ is the total number of voters.
\end{definition}
\noindent
The most well-known quota rule is the \emph{(weak) majority rule}, where $q=\frac{1}{2}$.

\medskip

It is worth noting that in the current setting we are often assuming that agents submitting their graphs have a well defined hierarchy of values in mind. Indeed, we can often safely assume that agents submit graphs which are \emph{justifiable}. A profile is justifiable if there is a single \textit{VAF} such that all members of the profile are defeat graphs of \textit{VAF}.\footnote{Justifiability of profiles of graphs has been studied in depth in \cite{airiau2016rationalisation}.}

\begin{definition}[Justifiable Profiles]
Let \textbf{AF} be a profile of graphs. \textbf{AF} is \emph{justifiable} if there is a \textit{VAF} such that for any \textit{AF}$_i \in \textbf{AF}$, \textit{AF}$_i$ is a single defeat graph of \textit{VAF}.
\end{definition}

In the later part of the paper we will mainly focus on graph aggregation rules restricted to justifiable inputs. In such rules, input profiles are only limited to profiles of graphs such that there is a \textit{VAF} such that all members of the profile are defeat graphs of \textit{VAF}.

We will be searching for graph aggregation rules which guarantee that if all the graphs considered in an input are defeat graphs of some \textit{VAF}, then so is the output. 
It is worth noting that this problem is a special case of the problem of \emph{collective rationality}, studied intensively in the social choice literature, also in the context of the abstract argumentation (see, e.g \cite{chen2019preservation, rahwan2010collective}). This problem relates to the issue of whether an aggregation rule makes sure that if  all of the agents involved in a decision process are submitting an option which satisfies a certain rationality constraint, so does the outcome of the procedure. In the context of graph aggregation, collective rationality means that if all input graphs satisfy a given property, the output graph satisfies it as well.

We will refer to the collective rationality with respect to the property described before as to the \emph{preservation of being a defeat graph}.

\begin{definition}[Preservation of being a defeat graph]
    A graph aggregation rule preserves being a defeat graph of \textit{VAF} if whenever for every $AF_i \in \textbf{AF}$, \textit{AF}$_i$ is a defeat graph of \textit{VAF}, so is $F(\textbf{AF})$.
    
\end{definition}

We can now illustrate the application of graph aggregation rules on the running example. 

\begin{example}\label{ex:graphaggr}
(Continuation of Example~\ref{ex:borda}). Consider now the defeat graphs induced by the experts' audiences as the input of the aggregation process. We will apply the majority rule to provide a common graph for the panel. Let us list out the attacks such that the majority of agents agree that they should hold: $F \rightarrow G$, $F \rightarrow C$, $F \rightarrow E$, $C \rightarrow B$, $E \rightarrow D$, $B \rightarrow A$. 

Note that
these edges corresponds to the graph obtained while using the Borda rule in the Example~\ref{ex:borda}. It is worth observing, however, that it is not necessarily always the case, as will be 
highlighted by the comparison between graph and preference aggregation conducted in the following sections. 
\end{example}

\subsection{Aggregation Axioms}

In following sections we will often refer to \emph{desirable properties} of the aggregation mechanisms under consideration. Here, we will define such properties (or \emph{axioms}), both for preference and graph aggregation. The considered axioms are standard in computational social choice, and are explained for instance in \cite{brandt2016handbook} for preference aggregation and \cite{endriss2014collective} for graph aggregation. The wording of the definitions has been changed for the ease of presentation.

\subsubsection{Preference aggregation rule}

Let us start with defining the desired properties for the preference aggregation approach. We will start with an informal description of the axioms, which will be followed by a formal definition.

A preference aggregation function is \emph{unanimous} if it never chan\-ges any ordering between options that all agents agree upon; the function is 
    \emph{anonymous} if it provides the same output regardless of the ordering of items in its input; 
    \emph{independent}, if the decision about the ordering of two values only depends on the way in which voters order this particular pair.
In addition, we will strive to find rules which are not \emph{dictatorial}. Formally:


\begin{definition}
A preference aggregation function $F$ is:
 \begin{itemize}
     \item \emph{Unanimous}: if whenever in a profile of orderings $\textbf{P}$=$\langle P_1, \dots, P_n \rangle$ all voters submit that $v_i \succ v_j$, then $v_i \succ v_j$ in $F(\textbf{P})$.
     \item \emph{Anonymous}: if for any profile of orderings \textbf{P}=$\langle P_1, \dots, P_n \rangle$, any pair of items $v_i, v_j$ and any permutation $\pi: \mathcal{N} \rightarrow \mathcal{N}$, $a \rightarrow b \in F(\textbf{P})$ iff $a \rightarrow b \in F(P_{\pi(1)}, \dots, P_{\pi(n)})$.

     \item \emph{Independent}: if it holds that for any pair of profiles of preference orderings $\textbf{P}$, $\textbf{P}'$  and any pair of values $v_1, v_2 \in V$, if $N_{\textbf{P}}^{v_1 \succ v_2} = N_{\textbf{P'}}^{v_1 \succ v_2}$, then $ v_1 \succ v_2  \in F(\textbf{P})$ iff $ v_1 \succ v_2  \in F(\textbf{P'})$.
     \item \emph{Dictatorial}: if there is $i$ such that for any \textbf{P}=$\langle P_1, \dots, P_n \rangle$, $F(\textbf{P})=P_i$.
 \end{itemize}
\end{definition}

\subsubsection{Graph aggregation rule}

Let us now define corresponding axioms for graph aggregation. They are adopted from \cite{endriss2014collective}. 

The \emph{unanimity} axiom states that if all agents agree that some attack should be included in the collective graph, then it is.
The \emph{anonymity} condition expresses that a choice of attacks does not depend on the name of voters.
\emph{Independence} states that all attacks are treated equally in any profile of defeat graphs.
We also additionally demand for a rule not to be \emph{dictatorial}. Formally:


\begin{definition}

A graph aggregation rule $F$ is:

    \begin{itemize}
        \item \emph{Unanimous:} if for every profile of argumentation frameworks \textbf{AF}=$\langle AF_1, \dots, AF_n \rangle$, if there is some pair of arguments $a,b \in A$ such that for every $AF_i \in \textbf{AF}$ $a \rightarrow_i b$, $a \rightarrow b \in F(\textbf{AF})$.

        \item \emph{Anonymous: }if for every profile of argumentation frameworks \textbf{AF}=$\langle AF_1, \dots, AF_n \rangle$, any attacks $a \rightarrow b$ and any permutation $\pi: \mathcal{N} \rightarrow \mathcal{N}$, $a \rightarrow b \in F(\textbf{AF})$ iff $a \rightarrow b \in F(\textbf{AF}_{\pi(1)}, \dots, \textbf{AF}_{\pi(n)})$.

        \item \emph{Independent: }if for any pair of profiles of argumentation frameworks $\textbf{AF}, \textbf{AF'}$, if $N_{a \rightarrow b}^{\textbf{AF}} = N_{a \rightarrow b}^{\textbf{AF'}}$, then $a \rightarrow b \in F(\textbf{AF})$ iff $a \rightarrow b \in F(\textbf{AF'})$.

         \item \emph{Dictatorial}: if there is $i$ such that for any \textbf{AF}, $F(\textbf{AF})=\textit{AF}_i$.
    \end{itemize}
  
\end{definition}

\section{Impossibility Results}

In this section we explore the limitation of the two aggregation approaches defined in Section~2, thus comparing aggregation at the level of values with aggregation at the level of attack graphs in the context of value-based argumentation.

\subsection{Graph aggregation}

We will commence with exploring the properties of the graph aggregation in the context of aggregating views on rankings of values. 

Note that graph aggregation requires substantially less information than the preference aggregation based mechanism. Indeed, a graph aggregation rule can be used without specifying a context of a \textit{VAF}. On the other hand, a preference aggregation can only be obtained when the values to which arguments appeal are specified. So, to use preference aggregation in the context of value-based argumentation we need to have the knowledge of both the profile of graphs and the \textit{VAF}.

Unfortunately, as graph aggregation rules do not take the context of a \textit{VAF} into account it might be the case that even though all agents participants of a discussion submit argumentation frameworks which are induced by their preferences over values, the collective graph is not justifiable by any preference ordering in the context of a given debate based on values. This is why the preservability of being a defeat graph is of high interest. In this section we will investigate the conditions  which graph aggregation rules need to satisfy to preserve being a defeat graph.

Unfortunately, as we will see later, graph aggregation rules that also preserves being a defeat graph cannot satisfy all of the axioms we defined.
We begin by showing that no quota rule can preserve being a defeat graph:

\begin{proposition}\label{quotas}
	Being a defeat graph is not preserved by any quota rule.
\end{proposition}
\begin{proof}
	
	Consider a quota defeat aggregation rule $F$ with an arbitrary quota $q \in [0,1]$. Then, take some natural number $n$ such that $\frac{1}{n} < q$. Further, construct a $\textit{VAF}= \langle A, \rightarrow, V, val \rangle $ such that  $A = \{ a_1, \dots, a_n \}$, $\rightarrow= \{ a_i \rightarrow a_{i+1}| i <n \} \cup \{ a_n \rightarrow a_1 \}$. Note that this attack relation forms a cycle. Also, let $val(a_i)=v_i$ for any argument $a_i$ (now all arguments are assigned unique values). Then, consider a set of agents $N = \{1, \dots, n \}$, submitting defeat graphs such that for any $i < n$, in agent $i$'s perspective only $a_i \rightarrow_i a_{i+1}$, while for agent $n$ only $a_n \rightarrow^n a_1$. It is easy to see that these are defeat graphs of \textit{VAF}. For any agent $i$, the set of attacks $\{a_n \rightarrow a_m | a_n \not\rightarrow_i a_m \}$ is a chain of length $n-1$. Then, we can consider a preference ordering over values such that for any  $a_n \rightarrow a_m$ such that $a_n \not\rightarrow_i a_m$, $val(a_m) \succ_i val(a_n)$. Clearly, this gives us a desired defeat graph.

Note now, that in the result of application of $F$ to this profile, no attacks are preserved, as each of them has a support of fewer agents than $q * |N|$. But now suppose that we have an ordering $P$ over $V$ under which such a defeat graph would be obtained. Then, we would need to have than $v_n \succ_P v_{n-1} \succ_P \dots \succ_P v_1 \succ_P v_n$. But then $P$ is not transitive, so it is not a preference ordering.
\end{proof}

Let us consider the following axiomatic characterisation of the class of quota rules:

\begin{theorem}[\cite{endriss2017graph}] \label{quot}
A graph aggregation $F$ rule is anonymous, monotonic and independent iff $F$ is a quota rule. 
\end{theorem}

This leads us immediately to an impossibility result:

\begin{corollary}\label{GraphGeneral}
    Any graph aggregation rule preserving being a defeat graph violates anonymity, monotonicity, or independence.
\end{corollary}

\begin{proof}
    If a rule $F$ is anonymous, monotonic and independent, by Theorem \ref{quot} $F$ is a quota rule. But then, by Proposition \ref{quotas}, $F$ does not preserve being a defeat graph.
\end{proof}

By Corollary \ref{GraphGeneral} it follows that all graph aggregation rule preserving being a defeat graph must violate intuitive axioms \emph{in the general case}. However, we are especially interested in aggregating justifiable graphs. So, we would like to establish which rules restricted to justifiable inputs preserve being a defeat graph. As we will see, this is impossible for an important class of rules.

To establish this result we will employ Arrow's impossibility theorem, one of the cornerstones of social choice theory, stated here in its version for strict linear orders:

\begin{theorem}[\cite{arrow1952social}]\label{Arrow}
 Any unanimous and independent preference aggregation rule is dictatorial, when the set of alternatives has at least 3 elements.
\end{theorem}

We can now show the following:

\begin{theorem}
Any graph aggregation rule restricted to justifiable profiles and preserving being a defeat graph is not independent, not unanimous, or is dictatorial. 
\end{theorem}

\begin{proof}
Take a graph aggregation rule which preserves being a defeat graph. Also, suppose towards contradiction that it is unanimous, independent and is not dictatorial. Now, consider a \textit{VAF}=$\langle A, \rightarrow, V, val \rangle$ such that $V=\{v_1,v_2,v_3\}$, $A= \{a_{v_i}|v_i \in V \}$, $val(a_{v_i})=v_i$ and $\rightarrow = \{a \rightarrow b| a, b \in A \textnormal{ and } a \neq b \}$. Now consider an argumentation framework \textit{AF} which is a defeat graph of \textit{VAF}. Note, that it corresponds to a unique preference ordering over $V$. We can thus encode any profile \textbf{P} of preference orderings over $V$ as a profile of defeat graphs \textbf{AF}$_{\textbf{P}}$ of \textit{VAF}. 
Define now a preference aggregation rule $F': \mathcal{P}^m \rightarrow \mathcal{P}$, such that  $v_i \succ v_j \in F'(\textbf{P})$ iff $a_{v_i} \rightarrow_{F(\textbf{AF}_{\textbf{P}})} a_{v_j}$. Notice that $F'(\textbf{P})$ is a linear ordering: as $F(\textbf{AF})$ is a defeat graph of \textit{VAF}, we have that $F'(\textbf{P})$ is connected since for every pair $\langle v_i, v_j\rangle \in V^2$ there is an attack $a_{v_i} \rightarrow a_{v_j}$. It is also anti-symmetric, as for every pair of arguments $a_{v_i}, a_{v_j}$ we need to have that either $a_{v_i} \rightarrow_{F(\textbf{AF})} a_{v_j}$ or $a_{v_j} \rightarrow_{F(\textbf{AF})} a_{v_i}$ while it is not the case that $a_{v_i} \rightarrow_{F(\textbf{AF})} a_{v_j}$ and $a_{v_j} \rightarrow_{F(\textbf{AF})} a_{v_i}$. Otherwise it would be impossible to find a symmetric preference ordering \textbf{P'} such that $F(\textbf{AF})$ is a defeat graph of \textit{VAF} induced by \textbf{P'}. By a similar argument we can show that  we cannot have that $a_{v_i} \rightarrow_{F(\textbf{AF}} a_{v_j}$, $a_{v_j} \rightarrow_{F(\textbf{AF})} a_{v_k}$ and $a_{v_k} \rightarrow_{F(\textbf{AF})} a_{v_i}$. From this it follows that $F'(\textbf{P})$ is transitive.

We need to demonstrate that if $F$ is unanimous, independent and non-dictatorial, then so is $F'$. If $F$ is unanimous: suppose that $F'$ is not and take a profile of preference orderings \textbf{P} such that for some pair of values $v_i, v_j$, for every $P_k \in \textbf{P}$ $v_i \succ_k v_j$ but $v_j \succ_{F'(\textbf{P})}$. Then, by construction of $F$ we would have that for every member of the profile \textbf{AF} of graphs induced by \textbf{P} $a_{v_i} \rightarrow a_{v_j}$, but it would not be the case in $F(\textbf{AF})$, which contradicts the assumptions. If $F$ is independent, suppose that $F'$ is not and take profiles of orderings \textbf{P}, \textbf{P'} over $V$ such that for some $v_i, v_j$ we have that $v_i \succ v_j \in F'(\textbf{P})$ while  $v_j \succ v_i \in F'(\textbf{P'})$ even though $v_i \succ v_j$ is supported by the same voters in both profiles. Then, by construction of $F$ we would have that for the pair of profiles of graphs \textbf{AF}, \textbf{AF'} induced by \textbf{P}, \textbf{P'} $a_{v_i} \rightarrow_{F(\textbf{AF})} a_{v_j}$ but  $a_{v_j} \rightarrow_{F(\textbf{AF'})} a_{v_i}$, so $F$ wouldn't be independent.  Finally, if $F'$ would be dictatorial,  by construction of $F$ we would have that for some $i$, in any profile \textbf{AF}, $F(\textbf{AF})=\textit{AF}_i$, so by contraposition if $F$ is non-dictorial, so is $F'$. But this means that by Theorem \ref{Arrow} we cannot have independent,  unanimous and non-dictatorial rules restricted to justifiable profiles which preserve being a defeat graph. 

\end{proof}

This result shows that 
it is not possible to find any rule satisfying all the considered desiderata. However, we can still consider natural rules preserving being a defeat graph.
Note that any \emph{representative-voter} rule satisfies this property. These are the rules which always select some graph represented in the input profile (see e.g., \cite{EndrissGrandiAAAI2014}).  Another intuitive move would be to select an argumentation framework which is a defeat graph of a \textit{VAF} justifying the input profile and minimizing a distance from the submitted graphs, in line with techniques from belief merging cited in the related work section. This approach is however not guaranteed to work, and investigating this issue is an interesting direction for future work, as it might be the case that the output of such a rule is not a defeat graph of a second \textit{VAF'} which also happen to justify the input profile.

\subsection{Preference aggregation}

Let us now proceed to the study of properties of preference aggregation based mechanisms. These mechanisms have a major advantage over graph aggregation: they always produce justifiable outputs. 
However, the outcome of the preference aggregation based mechanism is dependent on the context of \textit{VAF}. By the context of a \textit{VAF} we mean the assignment of values to arguments.
Indeed, it is worth noting that a preference aggregation mechanism does not always ensure that if two profiles of preference orderings induce the same profile of defeat graphs, they will produce the same collective defeat graph. They might, in principle, provide different graphs depending on the context of \textit{VAF}. As we will see further, this is the case for all reasonable preference aggregation based mechanisms.

Formally, we will be looking for preference aggregation rules corresponding to some graph aggregation rule.

\begin{definition}[Corresponding graph aggregation]
    Let $F$ be a preference aggregation rule. A graph aggregation function $F'$ corresponds to $F$ if for every profile of justifiable argumentation frameworks, \textit{VAF}=$\langle A, \rightarrow, V, val \rangle$ and a profile of preference orderings \textbf{P} justifying \textbf{AF} with respect to \textit{VAF}, the defeat graph $\langle A, \rightarrow_{F(\textbf{P})}\rangle=F'(\textbf{AF})$. 
\end{definition}

A necessary condition for the existence of a graph aggregation rule corresponding to a preference aggregation rule $F$, is that given a \textit{VAF} and a profile \textbf{AF} of its defeat graphs, application of $F$ will result in the same collective defeat graph, no matter which profile of preference orderings justifying \textbf{AF} is chosen. Formally, we will refer to this property as to \emph{interpretation independence}.

\begin{definition}[Interpretation Independence]
A preference aggregation rule $F$ is \emph{interpretation independent} if for every \textit{VAF} and profile \textbf{AF} of defeat graphs of \textit{VAF}, we have that for every pair of profiles \textbf{P}, \textbf{P'} justifying \textbf{AF} with respect to \textit{VAF}, $\langle A, \rightarrow_{F(\textbf{P})} \rangle = \langle A, \rightarrow_{F(\textbf{P'})} \rangle$.
\end{definition}

The aforementioned fact follows from the following lemma:

\begin{lemma}\label{corr}
If there is a graph aggregation rule corresponding to a preference aggregation rule $F$, then $F$ is interpretation independent.
\end{lemma}

\begin{proof}
Proof by transposition. Take a preference aggregation rule $F$ which is not interpretation independent. Then take a profile \textbf{AF} of defeat graphs of a \textit{VAF} such that for two profiles  of preference orderings \textbf{P}, \textbf{P'}, $F$ produces different outcomes (\textit{AF}, \textit{AF'}). Then note that any graph aggregation function corresponding to $F$ would need to   have both \textit{AF} and \textit{AF'} as the output for \textbf{AF}, which is not possible.
\end{proof}

However, interpretation independence only holds if $F$ is independent.
\begin{proposition} \label{InterpretationIndependence}
A preference aggregation rule $F$ is interpretation independent iff $F$ is independent.
\end{proposition}

\begin{proof}
	$(\Leftarrow )$ Consider any  $\textit{VAF} = \langle A, \rightarrow, V, val \rangle$, as well as a profile of its defeat graphs $\bf{AF}$. Also, let $F$ be any independent preference aggregation rule. Further, take any pair of profiles $\textbf{P}, \textbf{P'}$ of preference orderings over $V$ inducing $\bf{AF}$. Now suppose that $\langle A, \rightarrow^{F(\textbf{P})} \rangle \neq  \langle A, \rightarrow^{F(\textbf{P'})} \rangle$. Without loss of generality assume that there is an attack $a \rightarrow b$ such that $a \rightarrow b \in \langle A, \rightarrow^{F(\textbf{P})} \rangle$ but $a \rightarrow b \notin \langle A, \rightarrow^{F(\textbf{P'})} \rangle$. Then, by connectedness we know that $val(a) \succ val(b) \in F(\textbf{P})$. Otherwise we would have that $val(b) \succ val(a) \in F(\bf{P})$, so the attack would be blocked. Then, take the set of voters $N_{\textbf{P}}^{val(a) \succ val(b)}$. Note that they must correspond to defeat graphs in which $a \rightarrow b$ is included. Other defeat graphs can only be justified with orderings in which $val(b) \succ val(a)$ and, by connectedness requirement, preservation of $a \rightarrow b$ needs to be justified with an ordering in which $val(a) \succ val(b)$. So, $N_{\textbf{P}}^{val(a)\succ val(b)}$ is also the set of supporters of $val(a) \succ val(b)$ in $\textbf{P'}$. So, by independence,  $val(a) \succ val(b) \in F(\textbf{P'})$. So, $a \rightarrow b \in \langle A, \rightarrow^{F(\textbf{P'})} \rangle$. Contradiction.
	
	$(\Rightarrow)$ Proof by transposition. Take a preference aggregation rule $F$ which is not independent. Let us show that there is some \textit{VAF} and a profile of its defeat graphs such that for two distinct profiles \textbf{P}$_1$, \textbf{P}$_2$ of preference orderings justifying it, the defeat graphs induced by $F(\textbf{P}_1)$ and $F(\textbf{P}_2)$ are not equal. We know that there is a pair of values $v_1, v_2$ and a pair of profiles of preference orderings $\textbf{P}_1, \textbf{P}_2$ such that $N^{v_1 \succ v_2}_{\textbf{P}_1} = N^{v_1 \succ v_2}_{\textbf{P}_2}$ but $v_1 \succ v_2 \in F(\textbf{P}_1)$ while $v_1 \succ v_2 \notin F(\textbf{P}_2)$. Take these profiles and construct a \textit{VAF}=$\langle A, \rightarrow, V, val \rangle$ such that $V$ is the set of values ordered by $\textbf{P}_1$ and $\textbf{P}_2$, $A= \{a_{v_k}| v_k \in V \}$, $\rightarrow = \{ a_{v_1} \rightarrow a_{v_2},  a_{v_2} \rightarrow a_{v_1} \}$ and for any $val(a_{v_k})=v_k$. Now, consider a profile \textbf{AF} of defeat graphs of \textit{VAF} induced by \textbf{P}$_1$. Note that both \textbf{P}$_1$ and \textbf{P}$_2$ justify \textbf{AF} since $N^{v_1 \succ v_2}_{\textbf{P}_1} = N^{v_1 \succ v_2}_{\textbf{P}_2}$. We know, however, that  $v_1 \succ v_2 \in F(\textbf{P}_1)$ while $v_1 \succ v_2 \notin F(\textbf{P}_2)$. Thus, we have that in the defeat graph induced by $F(\textbf{P}_2)$, $a_{v_2} \rightarrow a_{v_1}$, which is not the case in the graph induced by  $F(\textbf{P}_1)$ .
\end{proof}

It turns out that as a consequence of the Arrow's impossibility theorem, we can only guarantee that a preference aggregation rule is Interpretation Independent if it is not unanimous or is dictatorial.

\begin{theorem}\label{PrefBad}
The only unanimous preference aggregation rule corresponding to a graph aggregation function is a dictatorship.
\end{theorem}

\begin{proof}
    From Proposition \ref{InterpretationIndependence} and Theorem \ref{Arrow} we have immediately that any unanimous, interpretation independent preference aggregation rule is dictatorial. Then, by Lemma \ref{corr} we get that the only unanimous preference aggregation rule corresponding to a graph aggregation function is a dictatorship.
\end{proof}

On a more positive side, let us highlight two situations where preference aggregation rules are always interpretation independent.
The first such situation is when 
a given \textit{VAF} only admits defeat graphs justifiable by a \emph{unique} preference ordering. The following proposition characterizes such \textit{VAF}s.

\begin{proposition}\label{unique}
 For every \textit{VAF}$=\langle A, \rightarrow, V, val \rangle$ we have that for every defeat graph \textit{AF} of \textit{VAF}, \textit{AF} is justifiable by a unique preference ordering iff for every $v_1, v_2 \in V$ there is a pair of arguments $a_1, a_2$ such that $val(a_1)=v_1$, $val(a_2)=v_2$ and $a_1 \rightarrow a_2$ or $a_2 \rightarrow a_1$. 
\end{proposition}
\begin{proof}
$(\Rightarrow)$ Proof by transposition. Take a \textit{VAF}$=\langle A, \rightarrow, V, val \rangle$ such that for a pair of values $v_1, v_2 \in V$, for any $a_1, a_2 \in A$ if $val(a_1)=v_1$ and $val(a_2)=v_2$, $a_1 \not\rightarrow a_2$ and $a_2 \not\rightarrow a_1$. We will show that there are $\succ, \succ'$ such that $\langle A, \rightarrow_{\succ} \rangle = \langle A, \rightarrow_{\succ'} \rangle$. Let for every $v_1, v_2 \in V$ such that $v_1, v_2 \notin \{v_i, v_j \}$, $v_1 \succ v_2$ iff $v_1 \succ' v_2$. Also, let for every $v_1 \notin \{v_i, v_j \}$, $v_1 \succ v_i$, $v_1 \succ v_j$, $v_1 \succ' v_i$ and $v_1 \succ' v_j$. Now take any attack $a \rightarrow b$. Note that by construction of $\succ, \succ'$, if $val(a) \notin \{v_i, v_j \}$ or $val(b) \notin \{v_i, v_j \}$, $a \rightarrow_{\succ} b$ iff $a \rightarrow_{\succ'} b$. But if it is not the case, $a \not\rightarrow b$. So, $\langle A, \rightarrow_{\succ} \rangle = \langle A, \rightarrow_{\succ'} \rangle$.

$(\Leftarrow)$ Take a \textit{VAF}$=\langle A, \rightarrow, V, val \rangle$  such that for every $v_1, v_2 \in V$ there is a pair of arguments $a_1, a_2$ such that $val(a_1)=v_1$, $val(a_2)=v_2$ and $a_1 \rightarrow a_2$ or $a_2 \rightarrow a_1$. Then suppose that there are two preference orderings $\succ, \succ'$ such that $\langle A, \rightarrow_{\succ} \rangle \neq \langle A, \rightarrow_{\succ'} \rangle$. Consider a pair of values $v_1, v_2$ such that $v_1 \succ v_2$ but $v_2 \succ v_1$. W.l.o.g assume that there is a pair of arguments $a_1, a_2$ such that $val(a_1)=v_1$, $val(a_2)=v_2$ and $a_1 \rightarrow a_2$. But we know that $a_1 \rightarrow_{\succ} a_2$ iff $a_1 \rightarrow_{\succ'} a_2$, and hence $v_1 \succ v_2$ iff $v_1 \succ' v_2$, which contradicts the assumptions.
\end{proof}

Another situation guaranteeing interpretation independence is when a \textit{VAF} has not more than two values, as shown in the following proposition.

\begin{proposition}
 For every $\textit{VAF}=\{A, \rightarrow, V, val \}$ such that $|V| \leq 2$, either all defeat graphs are justified by a unique preference ordering, or all defeat graphs of \textit{VAF} are equal.
\end{proposition}

\begin{proof}
    Consider any such $\textit{VAF}=\{A, \rightarrow, V, val \}$ and let $V=\{v_1, v_2\}$ (note that if $V$ is a singleton, all defeat graphs of \textit{VAF} are equal by definition of a defeat graph). Suppose that for some pair of arguments $a_1, a_2$ such that $val(a_1)=v_1$ and $val(a_2)=v_2$, $a_1 \rightarrow a_2$ or $a_2 \rightarrow a_1$. Then, by Proposition \ref{unique} we have that all defeat graphs of \textit{VAF} are justified by a unique preference ordering. Also, if this is not the case, for all attacks $a \rightarrow b$ in $\rightarrow$, $val(a)=val(b)$. So, by definition of a defeat graph, for any preference ordering $\succ$, $\rightarrow = \rightarrow_{\succ}$.
\end{proof}

\section{Retrieved orderings}

Note that both of the approaches considered so far have some drawbacks. On the one hand, 
we have shown that any graph aggregation function satisfying desirable properties cannot ensure that a \textit{VAF} justifying a profile of graphs justifies also the collective graph. On the other hand, obtaining a collective defeat graph using preference aggregation functions requires more information than the employment of graph aggregation and we cannot ensure its correspondence to any graph aggregation rule. In this section we explore an aggregation method which combines the two approaches. Our goal will be to provide a procedure which both makes sure that the collective argumentation framework is always a defeat graph of the same \textit{VAF} as the input graphs, and that it corresponds to some graph aggregation rule.

To this purpose we consider the following mechanism. First, a profile of defeat graphs on a given \textit{VAF} is submitted by the agents. Then, a profile of preference orderings justifying the profile is derived (multiple choices are possible). Further, these preferences are aggregated using a suitable preference aggregation function, which in turns induces a collective defeat graph on the initial \textit{VAF}.

The considered mechanism is illustrated in Figure~\ref{scheme}. 
The lower tier of the picture represents a profile of defeat graphs of a certain \textit{VAF}. 
Each $\succ_i$ is a preference ordering inducing $\textit{AF}_i$, and is aggregated into $\succ_{coll}$ by means of a preference aggregation function $F$. This step is depicted as an arrow in the upper tier. The collective defeat graph $\textit{AF}_{coll}$ is then induced by \textit{VAF} using $\succ_{ coll}$.

\begin{figure}[H] 
\centering

\begin{tikzpicture}
		[=>,->,shorten >=1pt,auto,node distance=1.7cm,
		semithick]
		\node[shape=circle,draw=white] (A) {($\succ_1$,};
		\node[shape=circle,draw=white] (B) [ right of= A] {$\succ_2$,};
		\node[shape=circle,draw=white] (C) [ right of= B] {$\dots$};
		\node[shape=circle,draw=white] (D) [ right of= C] {$\succ_n$)};
		\node[shape=circle,draw=white] (E) [ right of= D] {$\Longrightarrow_F$};
		\node[shape=circle,draw=white] (F) [ right of= E] {$\succ_{coll}$};


		\node[shape=circle,draw=white] (M) [below of=A] {($\textit{AF}_1$,};
		\node[shape=circle,draw=white] (N) [below of= B] {$\textit{AF}_2$,};
		\node[shape=circle,draw=white] (O) [below of= C] {$\dots$};
		\node[shape=circle,draw=white] (P) [below of= D] {$\textit{AF}_n$)};
		\node[shape=circle,draw=white] (R) [below of= E] {};
		\node[shape=circle,draw=white] (S) [below of= F] {$\textit{AF}_{coll}$};

		\draw [->] (M) -- (A) node[midway,right] {\small \textit{VAF}} ;
		\draw [->] (N) -- (B) node[midway,right] {\small \textit{VAF}} ;
		\draw [->] (P) -- (D) node[midway,right] {\small \textit{VAF}} ;
		\draw [->] (F) -- (S) node[midway,right] {\small \textit{VAF}} ;
		

		\node[draw=black, fit=(A) (S) ](FIt1) {};
		\end{tikzpicture} 

\caption{The individual attack graphs $\textit{AF}_i$ are provided a justification order over values $\succ_i$, which are then aggregated using a preference aggregation rule $F$, and the collective ordering $\succ_{coll}$ so obtained induces a collective attack graph $\textit{AF}_{coll}$ justified by it.} \label{scheme}
\end{figure}
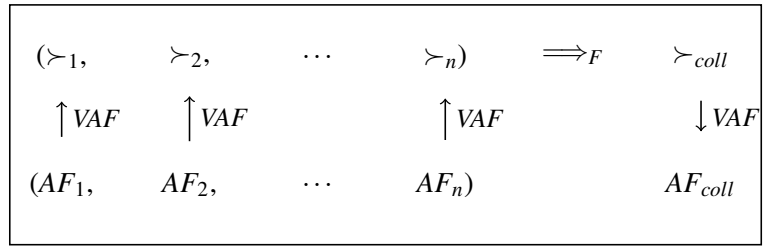
\vspace{-0.5em}

Let us define the proposed method formally.
Note that to ensure the functionality of the proposed mechanism, it is required that the selected choice of preference orderings is predetermined: we do not distinguish between exact preferences over values that agents believe in, as long as they induce the appropriate profile of defeat graphs. To cope with this issue, we will always specify a \emph{selection of justifications}. Intuitively, we infer preference orderings over values that participants of a debate have, based on their perceived strength of arguments.

\begin{definition}[Justification selection]
Let $(\textit{Graphs}^{\textit{VAF}})^n$ be the collection of all profiles of defeat graphs of \textit{VAF}=$\langle A, \rightarrow, V, val \rangle$ of length $n$, and \textit{Prefs}$^n$ be the set of all profiles of preference orderings over $V$ of length $n$. Then a \emph{justification selection} is a function $J: (\textit{Graphs}^{\textit{VAF}})^n \rightarrow \textit{Prefs}^n $ such that $J(\textbf{AF})$ induces~\textbf{AF} for any profile $\textbf{AF}$ in $(\textit{Graphs}^{\textit{VAF}})^n$ .
\end{definition}

Then, we can provide a definition of the proposed mechanism.

\begin{definition}[$C^J_F$-mechanism]
Take \textit{VAF}=$\langle A,\rightarrow, V, val \rangle$, a profile \textbf{AF} of defeat graphs of \textit{VAF}, a preference aggregation function $F$, and a justification selector $J$. Then, $C_F^J(\textit{VAF}, \textbf{AF})$  amounts to the defeat graph $\textit{AF}=\langle A, \rightarrow_{F(J(\textbf{AF}))}\rangle$.
\end{definition}

Henceforth, for the sake of clarity we will often assume that a justification selection is fixed and omit the superscript $J$.
Let us illustrate the mechanism on the running example.

\begin{example}
(Continuation of Example \ref{ex:graphaggr}) Consider the profile of graphs submitted by the experts as in Example~\ref{ex:graphaggr}. Note that they are all defeat graphs of the \textit{VAF} presented in Example \ref{ex:running}. Thus, we can retrieve the preference orderings justifying them and apply the Borda preference aggregation rule. Let us assume that the tie-breaking protocol follows a fixed ordering $ER > SF>IE>EV$. 
Finally, once the collective preference ordering has been obtained, we can construct the collective defeat graph.

Suppose that the selector $J$ associates with Expert~1 and Expert~3 the same orderings as in Example~\ref{ex:borda}, while for Expert~2 the selectore chooses $\textit{EV}\succ \textit{ER} \succ \textit{IE}\succ \textit{SF}$. If we aggregate the three preferences using the Borda rule we obtain $\textit{EV}\succ \textit{ER} \succ \textit{SF} \succ \textit{IE}$, which is a different result than the one in Example~\ref{ex:borda}. If we then construct the associated collective defeat graph, we observe that the attack from the argument $B$ on $A$ is blocked, while the remaining ones are as in the collective graph of Example~\ref{ex:borda}.  
\end{example}

The proposed mechanism $C_F$ enjoys a number of interesting properties. 
First, this method does not require information  about the exact preference orderings over values, so it does not require us to demand agents to provide additional information, as it was the case when the preference aggregation method was involved.
However, we still have not ensured that there is a graph aggregation rule corresponding to any mechanism $C_F$. To satisfy this, we would need to have that the collective graph is not dependent on the chosen \textit{VAF} justifying a profile of graphs.
\begin{definition}[\textit{VAF}-Independence]
 $C_F$ is \textit{VAF}-Independent if for any profile of justifiable graphs \textbf{AF} and a pair \textit{VAF}, \textit{VAF'} of VAFs justifying \textbf{AF},   $C_F(\textbf{AF}, \textit{VAF} )=C_F(\textbf{AF}, \textit{VAF'} )$.
\end{definition}

Unfortunately, as we will see, this can be the case only if $F$ violates unanimity or is dictatorial. 
\begin{proposition}\label{VafIndep}
If $C_F$ is \textit{VAF}-Independent, then $F$ is dictatorial or non-unanimous for $|V| \geq 4$.
\end{proposition}
\begin{proof}
    Let us firstly show that any \textit{VAF}-Independent $C_F$ requires $F$ to be independent for profiles based on at least 4 values. Suppose that $F$ is not independent. Then, take two profiles \textbf{P}, \textbf{P'} over a set $V$ with $|V| \geq 4$ such that for some pair of values $v_1$, $v_2$ we have that $N_\textbf{P}^{v_1 \succ v_2}=N_\textbf{P'}^{v_1 \succ v_2}$ but $v_1 \succ_{F(\textbf{P})} v_2$ while $v_2 \succ_{F(\textbf{P'})} v_1$. Then, consider a profile of graphs \textbf{AF}$= \langle \textit{AF}_1, \dots \textit{AF}_n \rangle$ such that $n= |\textbf{P}|$ and for $\textit{AF}_i \in \textbf{AF}$, $\textit{AF}_i = \{A, \rightarrow_i \}$ and $A= \{a_{v_i}| v_i \in V \}$, $\rightarrow_i = \{a_{v_1} \rightarrow_i a_{v_2} \}$ if $v_1 \succ_i v_2$, otherwise $\rightarrow_i = \{a_{v_2} \rightarrow_i a_{v_1} \}$. Now take two \textit{VAF}s such that \textbf{AF} is a profile of defeat graphs of \textit{VAF} justified by \textbf{P}, while \textbf{AF} is a profile of defeat graphs of \textit{VAF'} justified by \textbf{P'}, such that for both \textit{VAF}s, $a_{v_1}$ and $a_{v_2}$ attack each other, $val(a_{v_1})=val'(a_{v_1})=v_1$ and $val(a_{v_2})=val'(a_{v_2})=v_2$, while for some pair of arguments $a_{v_i}, a_{v_j}$, $val(a_{v_i})= v_i, val(a_{v_j})= v_j$ and $val'(a_{v_i})= v_j, val'(a_{v_j})= v_i$. Now note that the outcome of $C_F(\textbf{AF}, \textit{VAF})$, $C_F(\textbf{AF}, \textit{VAF'})$ only depends on the collective preference ordering over $v_1, v_2$. So, we will have that $a_{v_1} \rightarrow_{C_F(\textbf{AF}, \textit{VAF})} a_{v_2}$, while $a_{v_2} \rightarrow_{C_F(\textbf{AF}, \textit{VAF'})} a_{v_1}$. So, $C_F$ is not \textit{VAF}-Independent.
    It follows now from Theorem \ref{Arrow} that if $C_F$ is \textit{VAF}-Independent, then $F$ is dictatorial or non-unanimous.
    \end{proof}
    
    This leads us to conclude that when there are at least 4 values, if $F$ is non-dictatorial and unanimous, a two arguably basic properties, then $C_F$ does not correspond to any graph aggregation rule. 
 %
%
Note, however, that this issue is solved once we fix the \textit{VAF} justifying particular profiles of defeat graphs.

Let us first define a fixed selection of a \textit{VAF} for a profile of graphs.
 \begin{definition}[VAF selection]
 Let $(\textit{Graphs})_J^n$ be the collection of all justifiable profiles of argumentation frameworks of length $n$, and \textit{VAFs} be the set of all \textit{VAFs}. Then a \emph{VAF selection} is a function $S: (\textit{Graphs})_J^n \rightarrow \textit{VAFs} $ such that for any profile $\textbf{AF}$, $S(\textbf{AF})$ is such that for every $\textit{AF}_i \in \textbf{AF}$, $\textit{AF}_i$ is a defeat graph of $S(\textbf{AF})$.
 
 \end{definition}
 
 We can now define the modified mechanism.
 
\begin{definition}[$C^{VAF}_{F,S}$ mechanisms with fixed \textit{VAF}]
 Let $F$ be a preference aggregation rule. Then, the combined mechanism is the function $C_{F,S}^{\textit{VAF}}: Graphs_J \rightarrow Graphs$, where $Graphs_J$ is the set of all justifiable profiles of graphs, such that $C_{F,S}^{VAF}(\textbf{AF})=\langle A, \rightarrow_{F(\textbf{P})}\rangle$ such that \textbf{P} is a profile of preference orderings justifying  \textbf{AF} from the perspective of $S(\textbf{AF})$.

\end{definition}

In what follows we will assume that $S$ is fixed and drop the subscript when clear from context.
Then, it is not difficult to show that $C_F^{VAF}$ corresponds to a graph aggregation rule.
\begin{observation}
 Any $C_F^{VAF}$ corresponds to a graph aggregation rule restricted to justifiable inputs.
\end{observation}
\begin{proof}
    Take a profile of graphs \textbf{AF} justifable by a given \textit{VAF} and a combined mechanism  $C_F^{\textit{VAF}}$. Then notice that the choice of \textit{VAF} and of the profile of preference orderings justifying \textbf{AF} have been pre-determined, 
     and thus we can only obtain one graph. This is the case for all justifiable profiles of graphs. So, $C_F^{\textit{VAF}}$ indeed corresponds to a graph aggregation rule.
\end{proof}
Not only this movement ensures the existence of a corresponding graph aggregation rule, but also that the correspondent inherits beneficial properties of the used preference aggregation rule.
\begin{proposition}
Let $F$
be a preference aggregation rule. Then it holds that:
(1) If $F$ is unanimous, then the graph aggregation rule $F'$ corresponding to a $C_F^{\textit{VAF}}$

 is unanimous.
(2) If $F$ is anonymous, then the graph aggregation rule $F'$ corresponding to a $C_F^{\textit{VAF}}$
 is anonymous.

 is monotonic.
(3) If $F$ is independent, then the graph aggregation rule $F'$ corresponding to a $C_F^{\textit{VAF}}$
 is independent.
\end{proposition}
\begin{proof}

(1) Take any unanimous preference aggregation rule $F$. Then, suppose that the graph aggregation rule $F'$ corresponding to $C_F^{\textit{VAF}}$ is not unanimous. Then, take a \textit{VAF} and a profile of preference orderings $\textbf{P}^*$ inducing a profile of defeat graphs $\textbf{AF}$ such that there is an attack $a \rightarrow b$ such that for any $AF_i \in \textbf{AF}$, $a\rightarrow b \in AF_i$, but $a \rightarrow b \notin F'(\textbf{AF})$. Note that then $val(a) \neq val(b)$. Then, we must have that all agents submit that $val(a) \succ val(b)$. But then, by unanimity of $F$, $val(a) \succ val(b) \in F(\textbf{P}^*)$, and thus $a \rightarrow b \in F'(\textbf{AF})$. Contradiction.

(2) Take any anonymous preference aggregation rule $F$. Then, suppose that the graph aggregation rule $F'$ corresponding to $C_F^{\textit{VAF}}$ is not anonymous. Then, take a \textit{VAF} and a profile of preference orderings $\textbf{P}$ and a sequence $\textbf{AF}$ of defeat graphs induced by $\textbf{P}$ such that there is a permutation $\pi$ of $\textbf{AF}$ such that $F'(\textbf{AF}) \neq F'(\pi(\textbf{AF}))$. But now note that this would imply that $F(\textbf{P}) \neq F(\pi(\textbf{P'}))$ which cannot be the case by anonymity of $F$.

(3) Take any independent preference aggregation rule $F$. Then, suppose that the graph aggregation rule $F'$ corresponding to $C_F^{VAF}$ is not independent. So, take a \textit{VAF} and two profiles of defeat graphs \textbf{AF}, \textbf{AF'} of \textit{VAF} such that there is some attack $a\rightarrow b$ such that $N^{a \rightarrow b}_{\textbf{AF}} = N^{a \rightarrow b}_{\textbf{AF'}}$, but $a \rightarrow b \in F'(\textbf{AF})$ while $a \rightarrow b \notin F'(\textbf{AF'})$. Note that by connectedness of preference orderings, for any justification $\textbf{P}, \textbf{P'} $ of \textbf{AF}, \textbf{AF'}, $N^{val(a) \succ val(b)}_{\textbf{P}} = N^{val(a) \succ val(b)}_{\textbf{P'}}$. So, by independence, $a \succ b \in F(\textbf{P})$ iff $a \succ b \in F(\textbf{P})$. So, if  $a \rightarrow b \in F'(\textbf{AF})$, then $a \rightarrow b \in F'(\textbf{AF'})$. Contradiction.
\end{proof}

So, we have provided an alternative method of aggregating agents' views on the importance of values in the context of value-based argumentation. This has allowed us to overcome the issues we pointed out for the preference aggregation and graph aggregation approaches: we ensure that the outcome of the procedure is a defeat graph of the considered \textit{VAF}, while making sure that the mechanism corresponds to a graph aggregation rule restricted to justifiable inputs. 

\section{Conclusion and Perspectives}

In this paper we have explored three plausible methods of aggregating agents' views in the context of value-based argumentation, in order to obtain a collective argumentation framework: preference-based, graph-based, and combination-based aggregation methods. 

For the preference aggregation approach, we have shown that non-dictatorial rules never correspond to graph aggregation rules and thus, depending on the context, different collective argumentation frameworks can be obtained even though agents perceive the defeat relation in the same manner. 
On the other hand, 
no graph aggregation approach satisfying intuitive properties can ensure that the  outcome of the procedure is a defeat graph of the same \textit{VAF} as the input graphs. These issues can be circumvented by an aggregation approach combining preference aggregation with the extraction of an ordering over values from the individual attack relations, which fixes a \textit{VAF} and devises a suitable aggregation procedure for the specific context. 
This is 
one of the main conclusions of this work: obtaining a \emph{justifiable} collective argumentation graph is not possible without taking into consideration the underlying \textit{VAF}, which acts as the context of the problem.

This observation leaves a vast room for further research, investigating further aggregators that are specific to a given (class of) \textit{VAF}. In the current setting we have restricted ourselves to studying preference orderings over values expressed as \emph{linear} orderings. This is a strong requirement. It would be of high interest to study the aggregation problems stated in this paper when agents are allowed to consider particular values as equally important. Moreover, it would be interesting to study the scenario in which agents do not agree on the assignment of values to arguments. Further, it would be of natural interest to investigate, following e.g. Kaci and van der Torre \cite{kaci2008preference}, the scenario in which arguments appeal to multiple values. 
Finally, studying classes of argumentation problems which are not affected by the problems highlighted in the paper would be very beneficial towards showing the applicability of the proposed approaches.

\subsubsection*{Acknowledgments}
We would like to thank Sonja Smets for useful comments in the early phases 
of this project. Most of the work was completed when Grzegorz Lisowski was hosted at the University of Toulouse in 2018.
\bibliographystyle{eptcs}
\bibliography{bibliography}
\end{document}